\documentclass[preprint]{elsarticle}
\usepackage{amssymb}
\usepackage{amsthm}
\newtheorem{theorem}{Theorem}
\newtheorem{definition}{Definition}

\journal{Information Processing Letters}
\begin{document}

\begin{frontmatter}
\title{Sensitivity versus block sensitivity of Boolean functions}
\author{Madars Virza}
\ead{madars.virza@lu.lv}
\address{University of Latvia, Rainis bld. 19, Riga, LV-1586, Latvia}

\begin{abstract}
Determining the maximal separation between sensitivity and block sensitivity of Boolean functions is of interest for computational complexity theory. We construct a sequence of Boolean functions with $bs(f) = \frac{1}{2}{s(f)}^2+\frac{1}{2} s(f)$. The best known separation previously was $bs(f) = \frac{1}{2} {s(f)}^2$ due to Rubinstein. We also report results of computer search for functions with at most $12$ variables.
\end{abstract}

\end{frontmatter}

\section{Introduction}

In his 1989 paper \cite{nisan} Noam Nisan gives tight bounds for computing the value of a Boolean function in CREW-PRAM model. These bounds are expressed in terms of two complexity measures, namely the sensitivity and block sensitivity.

Let $w$ be a Boolean string of length $n$ and let $S$ be any subset of indices. Following definitions in \cite{nisan}, by $w^{S}$ we will mean $w$ with all bits in $S$ flipped. If $f(w) \neq f(w^{S})$, we will say that $f$ is \textit{sensitive to $S$ on $w$}.

\begin{definition}
    The \textit{sensitivity} $s(f, w)$ of $f$ on the input $w$ is defined as number of indices $i$ such that $f(w) \neq f(w^{\{i\}})$, $s(f, w) = | \{ i : f(w) \neq f(w^{\{i\}}) \} |$. The \textit{sensitivity} $s(f)$ of $f$ is $\displaystyle\max_w s(f, w)$.
\end{definition}

\begin{definition}
    The \textit{block sensitivity} $bs(f, w)$ of $f$ on n inputs $w_1, \ldots, w_n$ ($w = w_1 \ldots w_n$) is defined as maximum number of disjoint subsets $B_1, \ldots, B_k$ of $\{ 1, 2, \ldots, n \}$ such that for each $B_i$, $f(w) \neq f(w^{B_i})$. The \textit{block sensitivity} $bs(f)$ of $f$ is $\displaystyle\max_w bs(f, w)$.
\end{definition}

Block sensitivity is polynomially related to many other complexity measures such as decision tree complexity, certificate complexity, and quantum query complexity, among others \cite{buhrman}. The prime open problem associated with block sensitivity is whether sensitivity and block sensitivity are polynomially related.

In 1992 Rubinstein proposed \cite{rubinstein} a Boolean function $f$ for which $bs(f) = \frac{1}{2} {s(f)}^2$, demonstrating a quadratic separation between these two complexity measures, which we quote in verbatim:

\begin{theorem}[Rubinstein]
For every $n$ that is an even perfect square, we can construct a Boolean $f$ on $n$ variables with $2 bs(f)={s(f)}^2=n$.

Let $\Delta_i$ denote the interval $\Delta_i = \{ (i-1) \sqrt{n}+1, \ldots, i \sqrt{n} \}$ ($i = 1, \ldots, \sqrt{n}$).

Let $g_i$ denote the Boolean function defined as follows: $g_i(H) = 1$ exactly if $H \cap \Delta_i = \{ 2j-1, 2j \}$ for some $j$ such that $2j \in \Delta_i$.

We define $f$ to be join of all such $g_i$: $f(H) = g_1(H) \wedge \ldots \wedge g_{\sqrt{n}}(H)$.
\end{theorem}

It is believed that this separation is not far from optimal. However, the best known upper bound of block sensitivity in terms of sensitivity is exponential \cite{lblock}.

\section{Our result}
In this paper we give an improvement on Rubinstein's example, in the case where the input length $n$ is an odd square. Our function family is a modification of Rubinstein's, and our analysis is modelled on his.

\begin{theorem}
For every non-negative integer $k$ there exists a Boolean function $f$ of $n = (2k+1)^2$ variables, for which $s(f) = 2k+1$ and $bs(f) = (2k+1)(k+1)$.
\end{theorem}
\begin{proof}
We will divide variables into $2k+1$ disjoint sections with $2k+1$ variables in each section. We define $f$ to be $1$ iff there is a section $x_1, x_2, \ldots, x_{2k+1}$ such that either:
\begin{enumerate}[(i)]
    \item $x_{2i-1} = x_{2i} = 1$ for some $1 \leq i \leq k$ and all other $x_j$'s are $0$, or
    \item $x_{2k+1} = 1$ and all other $x_j$'s are $0$.
\end{enumerate}

We will call such section a ``good'' section. For any $i$, we will call $x_{2i-1}$ and $x_{2i}$ a ``pair''.

We observe that for input $w = 0 \ldots 0$ we have $s(f, w) = 2k+1$ and $bs(f, w) = (2k+1)(k+1)$. We will now prove that these are extremal values for sensitivity and block sensitivity, respectively.

Suppose we have an arbitrary input $w$. We will consider two cases:
\begin{enumerate}
    \item $f(w) = 0$. Then we claim that for each of $2k+1$ sections there is at most one bit whose change could flip the value of $f$.
    
    We will call a pair ``incomplete'', if exactly one of its bits is set to $1$. We consider three cases based on the number of incomplete pairs in a section:
    
    \begin{enumerate}
        \item there are at least two ``incomplete'' pairs. Then we can't change the value of $f$ by flipping just one bit, because doing so will leave at least one ``incomplete'' pair.
        \item there is exactly one ``incomplete'' pair $P$. We first note that we can't make the section ``good'' without flipping a bit in $P$. If unsetting the $1$-bit in $P$ makes the section ``good'', then $x_{2k+1} = 1$ or $x_{2i-1} = x_{2i} = 1$ for some $1 \leq i \leq k$ and setting the $0$-bit in $P$ will not make the section ``good''. If setting the $0$-bit in $P$ makes the section ``good'', then the only $1$-bit in $x$ is also in $P$ and unsetting it will not make the section ``good''. In either case there is at most one choice for the bit to alter.
        \item there are no ``incomplete'' pairs. Then flipping any of $x_i$ for $1 \leq i \leq 2k$ will introduce an ``incomplete'' pair and the section will not become ``good''. Therefore there is at most one bit whose change could flip the value of $f$, namely, $x_{2k+1}$.
    \end{enumerate}
    
    \item $f(w) = 1$. If there are two ``good'' sections, we can't change the value of $f$ by flipping just one input bit. If there is only one ``good'' section, we have at most $2k+1$ choices for the bit to alter.
\end{enumerate}
This proves that $s(f) \leq 2k+1$ and we indeed have $s(f) = 2k+1$.

We will now prove that $bs(f) \leq (2k+1)(k+1)$. Assume that the maximal block sensitivity is achieved using $u$ blocks of size $1$ and $v$ blocks of size at least $2$. Number of blocks of size $1$ can't exceed the sensitivity of the function, therefore $u \leq 2k+1$. The total size of all blocks is at most $u + 2v$, furthermore, the total size of all blocks can't exceed the total number of variables, therefore $u + 2v \leq {(2k+1)}^2$. Taking these two inequalities together we obtain $bs(f) = u + v = \frac{1}{2}(u + (u + 2v)) \leq \frac{1}{2} ((2k+1) + (2k+1)^2) = (2k+1)(k+1)$.
\end{proof}

This implies that for our function we have $bs(f) = \frac{1}{2}{s(f)}^2+\frac{1}{2} s(f)$. Our function improves an inequality from an open problem in \cite{lblock}.

\section{Maximal separation for functions of few variables}

We discovered our function by performing exhaustive computer search of all Boolean functions with at most $12$ variables.

The number of $12$-variable Boolean functions is $2^{4096}$ and a brute-force approach is clearly not feasible. In our experiments we reduced the block sensitivity problem to the SAT problem and used a SAT solver on the resulting problem instances. We built our SAT instances by considering $2^{12}$ variables corresponding to the values of $f(x_1, x_2, \ldots x_{12})$ and additional variables and clauses to describe constraints $s(f) \leq s$ and $bs(f) \geq bs$, for arbitrary constants $s$, $bs$. The SAT solver we used was cryptominisat \cite{cryptominisat} and it took us about a week of computing time.

\subsection{The constraint $bs(f) \geq bs$}
In order to describe $bs(f) \geq bs$ we constructed one SAT instance for each partition of $n$ into $bs$ parts $p_1, \ldots, p_{bs}$. The semantic meaning for the constraint we will describe is ``there is a function which is sensitive on blocks of sizes $p_1, \ldots, p_{bs}$''. The constraint $bs(f) \geq bs$ is feasible if and only if there exists a partition for which the constraint described below is satisfiable.

We note that for any $T \subseteq \{1, \ldots, n \}$ and $g(x) = f(x^{T})$ we have $s(f) = s(g)$ and $bs(f) = bs(g)$. This permits us to assume, without loss of generality, that $f$ attains the maximal block sensitivity on input $0\ldots0$, furthermore we can assume that $f(0\ldots0) = 0$. If $bs(f) \geq bs$ then by reordering variables we can enforce $\neg f(0\ldots0) \wedge \bigwedge_{i}{f({{0\ldots0}^{P_i}})}$ where $\displaystyle P_i = \{ \sum_{j=1}^{i-1} p_j + 1, \sum_{j=1}^{i-1} p_j + 2, \ldots, \sum_{j=1}^{i} p_j \}$. This was our constraint. We used this constraint to enforce $bs(f) \geq bs$.

\subsection{The constraint $s(f) \leq s$}

We constructed the constraint $s(f) \leq s$ as a conjunction of constraints of the form $s(f, w) \leq s$, where $w$ ranged over all $2^n$ inputs.

We will now describe the setup for the constraint $s(f, w) \leq s$. Let $b_i$ be the variable denoting $f(w^{\{i\}})$. We used auxiliary variables and clauses to implement ``counting'' in unary. That is, we introduced new Boolean variables $a_1, \ldots, a_n$ and defined a set of constraints that force $a_1+a_2+\ldots+a_n=b_1+b_2+\ldots+b_n$ and $a_1 \geq a_2 \geq \ldots \geq a_n$, essentially sorting the $b_i$'s.

If $\{ a_i \}$ is a permutation of $\{ b_i \}$ sorted in descending order, then the constraint $s(f, w) \leq s$ can be implemented as $(\neg f(w) \wedge \neg a_{s+1}) \vee (f(w) \wedge a_{n-s})$. The set of $a_i$'s was built incrementally using a dynamic programming table $\{ c_{i,j} \}$ ($0 \leq i, j \leq n$), where $c_{i,j} = 1$ iff at least $j$ of $b_1, \ldots, b_i$ are $1$. Following the semantic meaning of $c_{i,j}$ we have the base conditions of $c_{i,0} = 1$ and $c_{i, j} = 0$ ($i < j$). Furthermore, the following recurrence holds: $c_{i,j} = (c_{i-1, j-1} \wedge b_i) \vee c_{i-1,j}$ ($1 \leq j \leq i \leq n$). This can be proved by considering two cases: if $b_i = 1$ then $c_{i,j} = c_{i-1,j-1} \vee c_{i-1, j} = c_{i-1,j-1}$, because $c_{i-1,j} \to c_{i-1,j-1}$, if $b_i = 0$ then $c_{i,j} = c_{i-1,j}$; both of these equalities agree with the definition of $c_{i,j}$. The set of $a_i$'s is exactly the last row of the matrix: $a_i = c_{i,n}$.

\subsection{Results}

We have summarized our results in Table \ref{senst}. Our experiments show that for all functions with at most 12 variables $bs(f)$ does not exceed $\frac{1}{2} {s(f)}^2 + \frac{1}{2} s(f)$. This might suggest that our example is indeed optimal. Due to the results of Kenyon and Kutin, the best possible separation that can be achieved between the sensitivity and the block sensitivity for blocks of size at most 2 is quadratic, so the new result is not far from optimal in such sense.

\begin{center}
\begin{table}
\begin{center}
\begin{tabular}{|ccc|ccc|ccc|ccc|}
\hline
$n$ & $s$ & $bs$ & $n$ & $s$ & $bs$ & $n$ & $s$ & $bs$ & $n$ & $s$ & $bs$ \\
\hline
$\geq 4$ & $2$ & $3$ &  $\geq 8$ & $4$ & $6$ &  $\geq 10$ & $4$ & $7$ &  $\geq 11$ & $9$ & $10$ \\
\hline
$\geq 5$ & $3$ & $4$ &  $\geq 8$ & $6$ & $7$ &  $\geq 10$ & $6$ & $8$ &  $12$ & $4$ & $8$ \\
\hline
$\geq 6$ & $4$ & $5$ &  $\geq 9$ & $3$ & $6$ &  $\geq 10$ & $8$ & $9$ &  $12$ & $6$ & $9$ \\
\hline
$\geq 7$ & $3$ & $5$ &  $\geq 9$ & $5$ & $7$ &  $\geq 11$ & $5$ & $8$ &  $12$ & $8$ & $10$ \\
\hline
$\geq 7$ & $5$ & $6$ &  $\geq 9$ & $7$ & $8$ &  $\geq 11$ & $7$ & $9$ &  $12$ & $10$ & $11$ \\
\hline
\end{tabular}
\caption{All permissable $(n, s, bs)$ triples for $n \leq 12$}
\label{senst}
\end{center}
\end{table}
\end{center}

\section{Acknowledgements}
I would like to thank Andris Ambainis for introducing me to this problem and subsequent helpful discussions. Computational resources were provided in part by AILab of Institute of Mathematics and Computer Science, University of Latvia. I am thankful to the anonymous reviewers for their comments, which substantially improved this paper (FIXME).

\bibliographystyle{elsarticle-num}
\bibliography{blocksens}

\end{document}